\documentclass[11pt]{article}

\usepackage{amsmath}
\usepackage{amssymb}
\usepackage{amsthm}
\usepackage{graphicx}
\usepackage{enumerate}
\usepackage{theoremref}
\usepackage{spconf}

\usepackage{caption}
\usepackage{bm}
\usepackage[dvipsnames]{xcolor}
\usepackage{caption}
\usepackage[font=small,labelfont=bf]{caption}
\usepackage{subcaption}
\usepackage{algorithmic}
\usepackage{algorithm}
\usepackage[hidelinks]{hyperref}
\usepackage{cleveref}
\usepackage{todonotes}
\usepackage{subcaption}
\usepackage{titlesec}
\usepackage{setspace}

\newtheorem{theorem}{Theorem}[section]

\newtheorem{conjecture}{Conjecture}
\theoremstyle{definition}
\newtheorem{definition}{Definition}

\graphicspath{{fig/}}

\def\fro{\textnormal{F}}
\def\circ{\textnormal{circ}}
\def\diag{\textnormal{diag}}
\def\ifft{\textnormal{ifft}}
\def\fft{\textnormal{fft}}

\begin{document}

\title{Three-dimensional  Signal Processing: A New Approach in Dynamical Sampling via Tensor Products}
\name{Yisen Wang$^{1}$, HanQin Cai$^{2}$,  Longxiu Huang$^{1}$}

\address{$^{1}$ Michigan State University\\$^{2}$ University of Central Florida
}
\date{}
\maketitle

\begin{abstract}
The dynamical sampling problem is centered around reconstructing signals that evolve over time according to a dynamical process, from spatial-temporal samples that may be noisy. This topic has been thoroughly explored for one-dimensional signals. Multidimensional signal recovery has also been studied, but primarily in scenarios where the driving operator is a convolution operator. In this work, we shift our focus to the dynamical sampling problem in the context of three-dimensional signal recovery, where the evolution system can be characterized by tensor products. Specifically, we provide a necessary condition for the sampling set that ensures successful recovery of the three-dimensional signal. Furthermore, we reformulate the reconstruction problem as an optimization task, which can be solved efficiently.  To demonstrate the effectiveness of our approach, we include some straightforward numerical simulations that showcase the reconstruction performance.
\end{abstract}

\section{Introduction}
The concept of dynamical sampling, first introduced in the works of \cite{lu2009spatial}, 
addresses the challenge of compensating for spatial sampling deficiencies by leveraging the temporal evolution of data during recovery \cite{aldroubi2013dynamical}. This method utilizes the time-dependent nature of signal evolution, driven by external forces, to enhance the quality of the collected samples, 
setting it apart from traditional static sampling.

Dynamical sampling allows for efficient data acquisition by sampling only a subset of data points in time and space. This is particularly useful in systems where collecting data from every point in space or at every moment in time is either costly or impractical. For instance, in large sensor networks or medical imaging, dynamical sampling can reduce the number of measurements needed while still allowing accurate reconstruction of the full signal.

Dynamical sampling has been extensively explored for one-dimensional signals, with studies like \cite{aldroubi2017dynamical, aldroubi2019frames, aldroubi2017iterative,huang2024robust}.  Particularly, for the scenario where the evolution operator is represented by a matrix \(A \in \mathbb{C}^{d \times d}\) and the signal, \(f \in \mathbb{C}^d\), is to be recovered,  \cite{aldroubi2017dynamical} provides necessary and/or sufficient conditions on the sampling set of indices \(\Omega \subseteq \{1, 2, \ldots, d\}\) and  the numbers $\{\ell_i\}_{i\in\Omega}$ such that $f\in\mathbb{C}^d$ can be recovered from the samples $\{A^jf(i):i\in\Omega,j=0,\cdots,\ell_i-1\}$. 

Although one-dimensional signals have been extensively studied, research on multi-dimensional signals remains relatively limited.  However, in industrial applications, the observed time-varying signals often involve multiple variables, highlighting the critical importance of studying multi-dimensional dynamical sampling. For example, in sensor networks used for environmental monitoring or industrial processes, data such as temperature, pressure, and humidity are collected over time across various spatial locations, forming a three-dimensional tensor \cite{vairamani2013}. Each dimension can represent spatial coordinates and time, highlighting the complexity of the data. To date, the primary research in this area has focused on signals evolving under convolution-driven operators in multi-dimensional settings \cite{aceska2015multidimensional}. This underscores a significant gap in the literature and highlights the need for further investigation into multi-dimensional dynamical sampling. 


In this work, we explore the dynamical sampling problem where the initial signal \( \mathcal{F} \) is in \( \mathbb{C}^{m \times p \times n} \), and evolves over time driven by the t-product of the tensor \( \mathcal{A} \in \mathbb{C}^{m \times m \times n} \).  Specifically,  the signal at time \( t \) is transformed according to:
\begin{equation}\label{eqn:evolution rule}
    \mathcal{F}_t = \mathcal{A}^t \ast \mathcal{F},
\end{equation}
where \( \ast \) denotes the t-product between two tensors \cite{kilmer2011factorization}, and \( \mathcal{A}^t \) represents the \( t \)-th power of \( \mathcal{A} \) under the t-product. We consider the spatio-temporal sampling data represented by the set \( \Psi = \{\mathcal{F}_t(i,j,k) : (i,j,k) \in \Omega\subseteq[m]\times[p]\times[n], t \in \{0\}\cup[T-1]\} \) with $[m]=\{1,2,\cdots, m\}$. The objective of this study is to identify the conditions on \( \Omega \) and \( T \) necessary to guarantee the reconstruction of the initial signal \( \mathcal{F} \) from \( \Psi \), and to formulate the reconstruction of \( \mathcal{F}\) as an optimization problem that can be efficiently solved.
 
    \subsection{Contributions}
Our main contributions are as follows:
\begin{itemize}
    \item We have established a necessary condition on $\Omega$ for the successful recovery of the initial three-dimensional signal from the given samples.
    \item We have transformed the reconstruction of the three-dimensional signal $\mathcal{F}$ from spatio-temporal samples into $p$ independent optimization problems.
    \item We have conducted several experiments to demonstrate the effectiveness of our method, determine the optimal total sampling time 
$T$, and verify our conjecture.
\end{itemize}

\section{Preliminaries}
In this section, we introduce the mathematical notations and concepts required for our study, with a focus on the t-product and other tensor operations.
\begin{definition}[t-product] The t-product of tensors $\mathcal{T}_1\in\mathbb{C}^{m\times p\times n}$ and  $\mathcal{T}_2\in\mathbb{R}^{p\times q\times n}$ is denoted by $\mathcal{T}_1\ast\mathcal{T}_2=:\mathcal{T}\in\mathbb{C}^{m\times q\times n}$  and can be defined by the following steps:\\
    \begin{itemize}
    
\item  $\widehat {\mathcal{T}_1 }= \fft(\mathcal{T}_1,[], 3)$, $\widehat {\mathcal{T}_2 }= \fft(\mathcal{T}_2,[], 3)$ 
\item $\widehat {\mathcal{T}} (:,:,k)$=$\widehat {\mathcal{T}_1} (:,:,k)\widehat {\mathcal{T}_2} (:,:,k)$
\item $\mathcal{T}:=\mathcal{T}_1\ast \mathcal{T}_2=\ifft(\widehat{\mathcal{T}},[],3)$.
\end{itemize}
\end{definition}
Apart from t-product, we also involve other products between the tensors for the signal recovery.
\begin{definition}\label{def:others}
Other products between tensor:
\begin{itemize}
    \item  Element-wise tensor product $\odot$: $\mathcal{T}=\mathcal{T}_1\odot \mathcal{T}_2$ for $\mathcal{T}, \mathcal{T}_1, \mathcal{T}_2 \in \mathbb{C}^{m\times p\times n}$, with $[\mathcal{T}]_{i,j,k}=[\mathcal{T}_1]_{i,j,k}[\mathcal{T}_2]_{i,j,k}$.
    \item  Tube-wise circular convolution $\circledast$: $\mathcal{T}=\mathcal{T}_1\circledast \mathcal{T}_2$ for $\mathcal{T}, \mathcal{T}_1, \mathcal{T}_2 \in \mathbb{C}^{m\times p\times n}$, with $[\mathcal{T}]_{i,j,:}=[\mathcal{T}_1]_{i,j,:}\ast [\mathcal{T}_2]_{i,j,:}$.
    \item  Frontal-slice-wise product $\bigtriangleup$: $\mathcal{T}=\mathcal{T}_1\bigtriangleup \mathcal{T}_2$ for $\mathcal{T}_1 \in \mathbb{C}^{m\times n\times p}$, $\mathcal{T}_2 \in \mathbb{C}^{n\times s\times p}$, $\mathcal{T} \in \mathbb{C}^{m\times s\times p}$, with $[\mathcal{T}]_{:,:,k}=[\mathcal{T}_1]_{:,:,k}[\mathcal{T}_2]_{:,:,k}$.
\end{itemize}
\end{definition}

\section{Main results}
\subsection{Necessary condition}
We have initially focused on the sampling set \(\Omega\) structured in a lattice form, specifically \(\Omega = I \times J \times [n]\), where \(I \subseteq [m]\) and \(J \subseteq [p]\). In this configuration, we present the following result:

\begin{theorem}
   Suppose \( \mathcal{F} \in \mathbb{C}^{m \times p \times n} \) and \( \mathcal{A} \in \mathbb{C}^{m \times m \times n} \). And suppose that the signal at time \( t \) follows the transformation specified in \eqref{eqn:evolution rule}. Then, the recovery of \( \mathcal{F} \) from \( \Psi \), with \( \Omega = I \times J \times [n] \) where \( I \subseteq [m] \) and \( J \subseteq [p] \), is not possible if \( J \neq [p] \).
\end{theorem}
\begin{proof}
    Given that \(\Omega = I \times J \times [n]\), the samples at time \(t\) can be represented as:
\begin{equation}\label{eqn:samplesatt}
    \mathcal{Y}_t = [\mathcal{I}_m]_{I,:,:} \ast \mathcal{F}_t \ast [\mathcal{I}_p]_{:,J,:}
\end{equation}
where \(\mathcal{I}_m\) is the \(m \times m \times n\) identity tensor \cite{kilmer2011factorization}. Utilizing the properties of the t-product and applying the discrete Fourier transformation on the third dimension of both sides of \eqref{eqn:samplesatt}, we obtain:
\begin{equation}\label{eqn:samplesInfrequency}
    [\widehat{\mathcal{Y}}_t]_{:,:,k} = [\mathbb{I}_m]_{I,:} \widehat{\mathcal{A}}_{:,:,k}^t [\widehat{\mathcal{F}}]_{:,J,k}
\end{equation}
where $\mathbb{I}_m$ stands for the $m\times m$ identity matrix. 
Given that the Fourier transformation is a unitary transformation, reconstructing \(\mathcal{F}\) from \(\Psi\) is equivalent to reconstructing \(\widehat{\mathcal{F}}\) from \(\widehat{\Psi} = \{\widehat{\mathcal{Y}}_t : t \in 0 \cup [T-1]\}\).

From \eqref{eqn:samplesInfrequency}, it is evident that the reconstructions of \([\widehat{\mathcal{F}}]_{:,j,k}\) are independent for different \(j \in J\). This implies that if \(J \neq [p]\), there will be some \(j \in [p] \setminus J\) for which \([\widehat{\mathcal{F}}]_{:,j,k}\) cannot be reconstructed from \(\widehat{\Psi}\). Consequently, \(\widehat{\mathcal{F}}\) cannot be fully reconstructed from \(\widehat{\Psi}\). The result of this theorem is thus established.
\end{proof}
Based on this result, we propose the following conjecture, which we intend to explore in our future work:
\begin {conjecture}\label{conj1}
A necessary condition for the successful recovery of \( \mathcal{F} \) from \( \Psi \) is that \( \bigcup_{(i,j,k)\in\Omega} \{j\} = [p] \).
\end {conjecture}
Although we do not provide a formal theoretical proof for this conjecture here, we conducted simulations to test our hypothesis. The results indicate that losing any vertical index from the second dimension leads to a failure in recovery, thereby supporting the conjecture.
\subsection{Method development for signal recovery}\label{sec:alg}
We now focus on reconstructing $\mathcal{F}$ from $\Psi$. If the samples sufficiently guarantee the reconstruction of $\mathcal{F}$, then the following optimization problem will yield a unique solution:
\begin{equation}\label{eqn:DS2opt}
\min_{\mathcal{X}}\sum_{t=0}^{T-1}\sum_{(i,j,k)\in\Omega}\|[\mathcal{A}^t\ast \mathcal{X}]_{i,j,k}-[\mathcal{F}_t]_{i,j,k}\|_{\fro}^2.
\end{equation}
For clarity, we introduce the following definitions:
Let $\mathcal{P}_{\Omega}(\cdot)$ denote the projection of a tensor onto the observed set $\Omega$ such that
\[
[\mathcal{P}_{\Omega}(\mathcal{T})]_{i,j,k}=\begin{cases}
    \mathcal{T}_{i,j,k}, & \text{if } (i,j,k)\in\Omega\\
    0, & \text{otherwise}
\end{cases}.
\]
According to \Cref{def:others}, we can reformulate \eqref{eqn:DS2opt} as:
\begin{equation}\label{eqn:optProj}
\min_{\mathcal{X}}\sum_{t=0}^{T-1} \|\mathcal{P}_{\Omega}(\mathcal{A}^t\ast \mathcal{X})-\mathcal{P}_{\Omega}(\mathcal{F}_t)\|_{\fro}^2.
\end{equation}
Consider that
\begin{equation}
\begin{aligned}
\mathcal{P}_{\Omega}(\mathcal{F}_t) &= \mathcal{P}_{\Omega}\odot\mathcal{F}_t \\
&= \ifft(\widehat{\mathcal{P}_{\Omega}}\circledast\widehat{\mathcal{F}_t},[],3)/n.
\end{aligned}
\end{equation}
Drawing on the properties of the t-product and inspired by \cite{liu2019low}, we can convert the least squares minimization problem \eqref{eqn:optProj} into a frequency domain version:
 \begin{equation}\label{eqn:optfrequency}
    \min_{\widehat{\mathcal{X}}\in\mathbb{C}^{m\times p\times n}}\sum_{t=0}^{T-1} \|\widehat{\mathcal{P}_{\Omega}} \circledast(\widehat{\mathcal{A}}^t\bigtriangleup \widehat{\mathcal{X}})/n-\widehat{\mathcal{P}_{\Omega}(\mathcal{F}_t)}\|_{\fro}^2. 
 \end{equation}
This problem can be decomposed into $p$ separate subproblems, one for each $j\in[p]$, where we solve:
\begin{equation}\label{eqn:optfrequency-sub}
\min_{[\widehat{\mathcal{X}}]_{:,j,:} }\sum_{t=0}^{T-1} \|[\widehat{\mathcal{P}_{\Omega}}]_{:,j,:} \circledast(\widehat{\mathcal{A}}^t\bigtriangleup [\widehat{\mathcal{X}}]_{:,j,:})/n-[\widehat{\mathcal{P}_{\Omega}(\mathcal{F}_t)}]_{:,j,:}\|_{\fro}^2.
\end{equation}
The goal is to achieve a minimal value of zero for each subproblem. To facilitate this, we construct the following system for each $t$ and $j$:
\begin{equation}\label{eq:optfre-eq}
A_3(j)A_1(t)x(j)=b(j,t)
\end{equation}
where 
\begin{equation}
x(j)=\begin{bmatrix}[\widehat{\mathcal{X}}]_{:,j,1};\cdots;[\widehat{\mathcal{X}}]_{:,j,n}\end{bmatrix}\in\mathbb{C}^{mn\times 1},
\end{equation}
\begin{equation}
b(j,t)=\begin{bmatrix}[\widehat{\mathcal{P}}_{\Omega}(\mathcal{F}_t)]_{:,j,1};\cdots;[\widehat{\mathcal{P}}_{\Omega}(\mathcal{F}_t)]_{:,j,n}\end{bmatrix}\in\mathbb{C}^{mn\times 1},
\end{equation}
and $A_1(t)$ and $A_3(j)$ are defined as:
\begin{equation}
A_1(t)=\begin{bmatrix}
    [\widehat{\mathcal{A}}]_{:,:,1}^t&&\\
    &\ddots&\\
    &&[\widehat{\mathcal{A}}]_{:,:,n}^t
\end{bmatrix}\in\mathbb{C}^{mn\times mn},
\end{equation}
\resizebox{0.5\textwidth}{!}{$
A_3(j)=\begin{bmatrix}
    \diag([\mathcal{A}_2(j)]_{1,1,:})& \diag([\mathcal{A}_2(j)]_{1,2,:})&\cdots& \diag([\mathcal{A}_2(j)]_{1,n,:})\\
    \diag([\mathcal{A}_2(j)]_{2,1,:})& \diag([\mathcal{A}_2(j)]_{2,2,:})&\cdots& \diag([\mathcal{A}_2(j)]_{2,n,:})\\
    \vdots&\vdots&\ddots&\vdots\\
    \diag([\mathcal{A}_2(j)]_{n,1,:})& \diag([\mathcal{A}_2(j)]_{n,2,:})&\cdots& \diag([\mathcal{A}_2(j)]_{n,n,:})
\end{bmatrix}$}
with $[\mathcal{A}_2(j)]_{:,:,\ell}=\circ([\widehat{\mathcal{P}_{\Omega}}]_{\ell,j,:})$. Solving this system will enable us to recover $x(j)$. Once all $x(j)$ values are obtained, they are combined and reshaped into a tensor $\widehat{\mathcal{X}}_{\textnormal{app}}$ of size $m\times n\times p$. The estimation of $\mathcal{F}$ is then set as $\mathcal{X}=\ifft(\widehat{\mathcal{X}}_{\textnormal{app}},[],3)$.

\subsection{Experiments}

To evaluate the performance of our proposed method, we conducted several simulations aimed at recovering the initial signal. The experiments are structured in three parts: the first part evaluates the overall recovery performance and the point-wise recovery of our algorithm, the second part focuses on determining the optimal value of the parameter $T$, which is crucial for effective signal recovery, the third part verifies the conjecture, demonstrating that to fully recover the signal, the union of the second dimension in our dataset must equal to $p$, i.e., the second dimension of the initial signal.
\subsubsection{Datasets}
To generate the synthetic datasets, we first create a random tensor of size \(20 \times 15 \times 5\) as the initial signal \(\mathcal{F}\) and another tensor of size \(20 \times 20 \times 5\) as the driven operator \(\mathcal{A}\). Using these, we produce signals at different times \(t\) by applying the operation \(\mathcal{F}_{t} = \mathcal{A}^t * \mathcal{F}\).  
Next, we generate another tensor \(\mathcal{P}_{\Omega} \in \{0,1\}^{20 \times 15 \times 5}\) of the same size as \(\mathcal{F}\) to denote the sampled locations. The entries of \(\mathcal{P}_{\Omega}\) are generated using a Bernoulli distribution, where an entry of \(1\) indicates the presence of a sample and \(0\) indicates its absence. The probability of \(1\) in \(\mathcal{P}_{\Omega}\) is set to the sampling rate \(\alpha\). 
Finally, we extract the spatio-temporal samples by generating \(\mathcal{P}_{\Omega}(\mathcal{F}_t)\) at various time points \(t \in \{0\} \cup [T-1]\), capturing samples across different locations and times.

\subsubsection{Recovery accuracy}
In the first set of experiments, we evaluated recovery accuracy. We set the maximum sampling time \( T \) to 5 and repeated the experiment 10 times to assess the stability of the method. As shown in \Cref{fig:recovery performance}, when the sampling rate reached 40\%, the relative error decreased to approximately \( 10^{-12} \), demonstrating the effectiveness of our approach.
\begin{figure}[ht]
    \centering
    \includegraphics[width=0.68\linewidth]{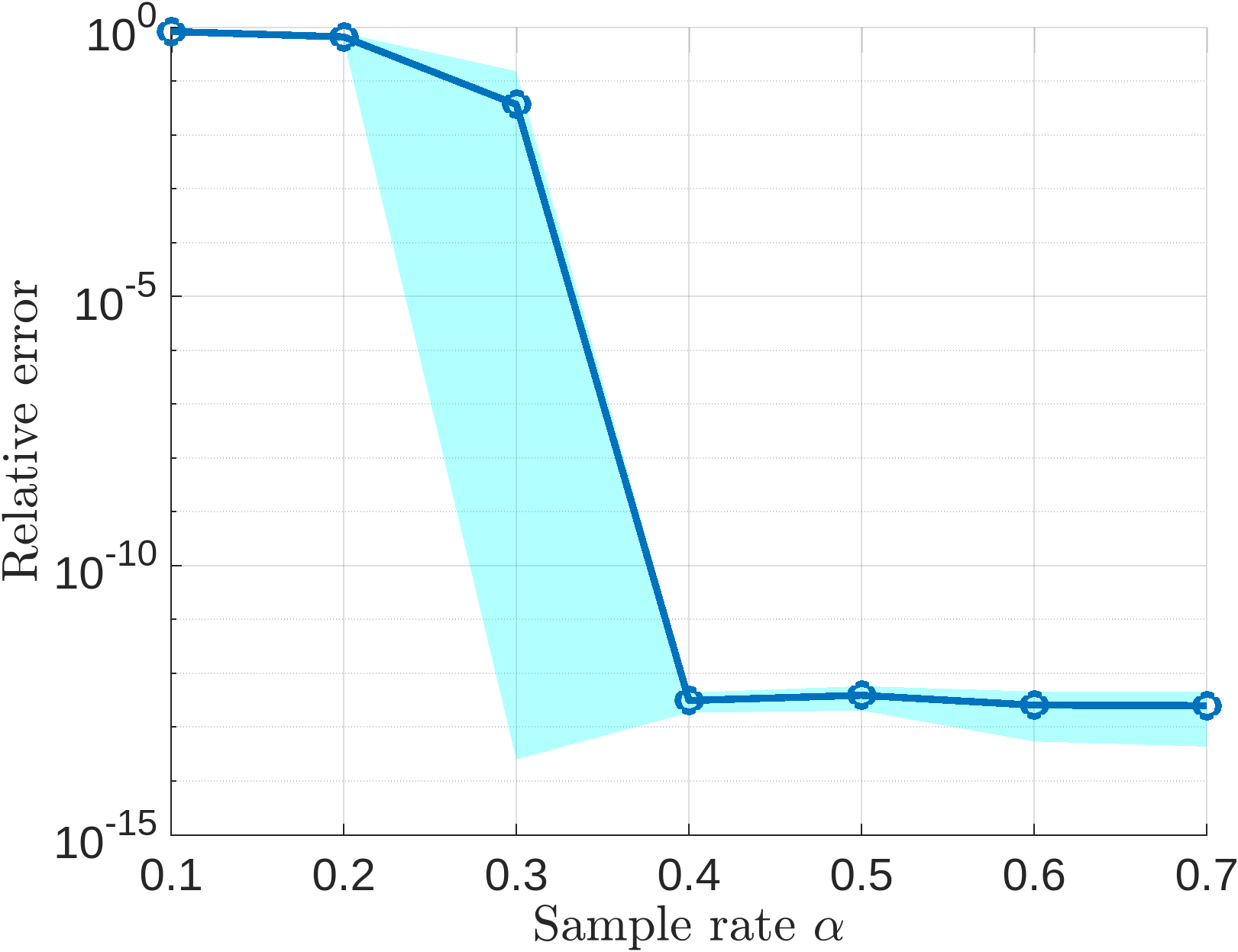}
    \caption{Relative error v.s. sampling rate $\alpha$: we repeated this experiment 10 times, calculating both the mean value and the standard deviation of the results. As shown by the shadow, the standard deviation is relatively small, indicating that our method consistently performs well across different trials.}
    \label{fig:recovery performance}
\end{figure}

\begin{figure}[ht]
    \centering
    \includegraphics[width=0.68\linewidth]{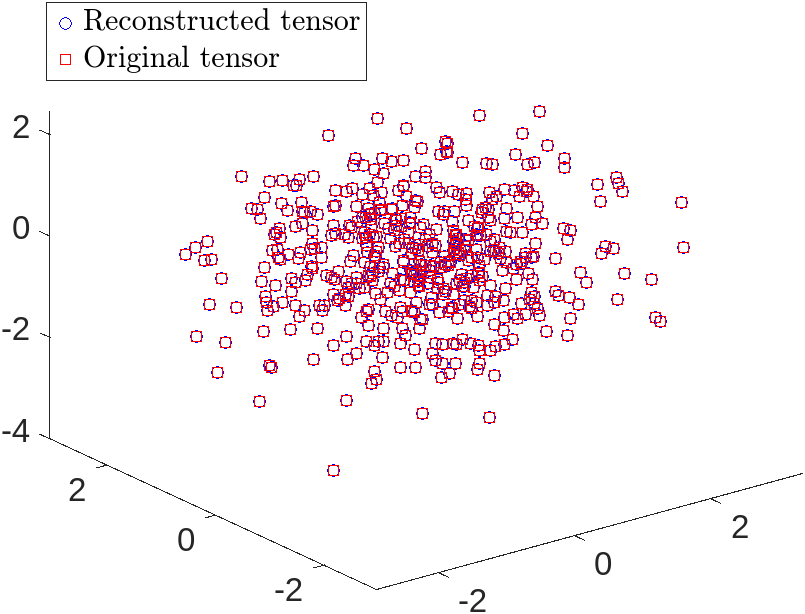}
    \caption{This figure shows the  point-wise gap between  
        the reconstructed signal and the ground truth signal, there are $20\times 15\times 5=1500$ sample points in total, we compared each point from the constructed tensor with the ground truth tensor.}
    \label{fig:point-wise recovery}
\end{figure}
 It is important to note that the product of \( T \) and the sampling rate \(\alpha\) provides a measure that reflects the overall sampling size.   If this product is less than 1, recovery is likely to fail, as the sample set lacks sufficient information to reconstruct the initial signal. Additionally, as observed in \Cref{fig:point-wise recovery}, a sampling rate of 40\% enables successful recovery of all points, further underscoring the robustness of our method.

\subsubsection{Optimal maximum sampling time T}
The parameter \( T \) is a critical hyperparameter, especially when samples are affected by noise. To investigate the effect of \( T \) on recovery performance, we conducted a series of experiments, varying \( T \) from 1 to 15 while keeping the sampling rate $\alpha$ fixed at 40\%. Additive Gaussian noises with mean 0 and variance \(\sigma^2\) were applied to the samples, i.e., \(\varepsilon \sim \mathcal{N}(0, \sigma^2)\).

\begin{figure*}[th]
    \centering
    \begin{subfigure}{0.32\textwidth}
        \centering
        \includegraphics[width=\textwidth]{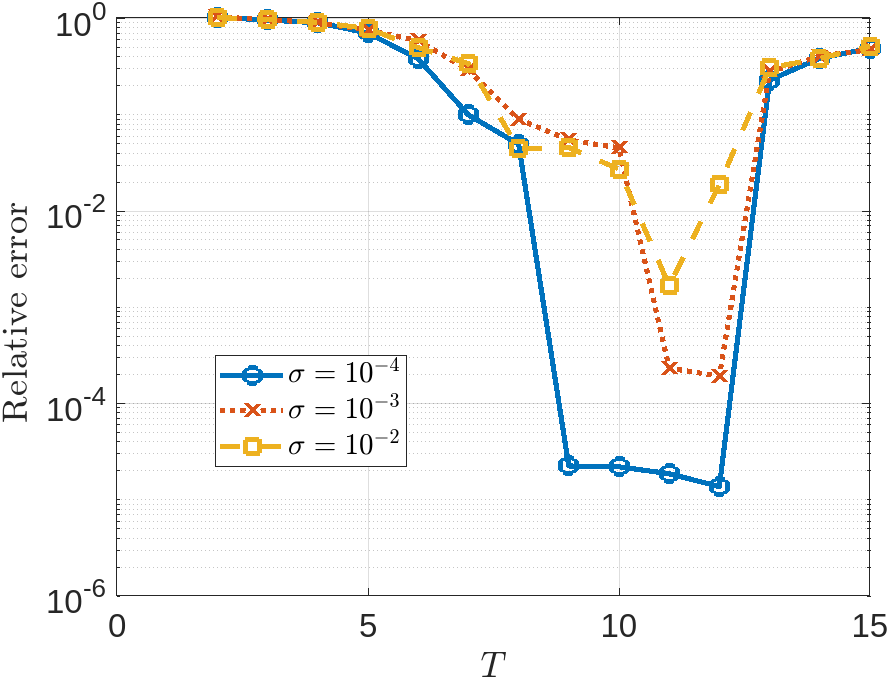} 
        \caption{$\alpha=0.2$}
    \end{subfigure}
    \begin{subfigure}{0.32\textwidth}
        \centering
        \includegraphics[width=\textwidth]{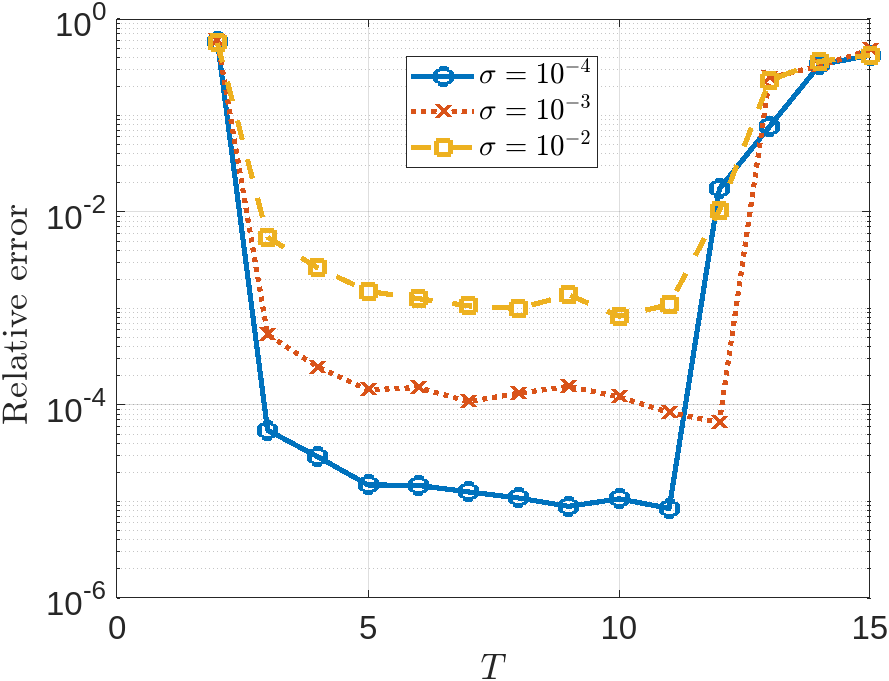} 
        \caption{$\alpha=0.5$}
    \end{subfigure}

    \caption{Relative error v.s. maximum sampling time  $T$   under different noise levels and sampling rates}
    \label{fig:Optimal T}
\end{figure*}
As shown in \Cref{fig:Optimal T}, increasing \( T \) does not always enhance recovery performance; rather, an optimal value of \( T \) exists. For \( T > 10 \), we calculate the condition number \(\kappa(j)\) of the matrix 
{\footnotesize
\[
\begin{bmatrix}
    (A_3(j)A_1(0))^\top & (A_3(j)A_1(1))^\top & \cdots & (A_3(j)A_1(T))^\top
\end{bmatrix}^\top
\]}
and define \( K = \max_{j} \kappa(j) \) as the condition number of the entire system. We observed that \( K \) became excessively large (on the order of \(10^{11}\)), which caused instability in the linear system \eqref{eq:optfre-eq} and led to a significant increase in relative error, as shown in \Cref{fig:condition number}. This indicates that selecting an appropriate value for \( T \) is crucial for achieving stable and accurate recovery.

 \begin{figure}[ht]
    \centering
     \includegraphics[width=0.6\linewidth]{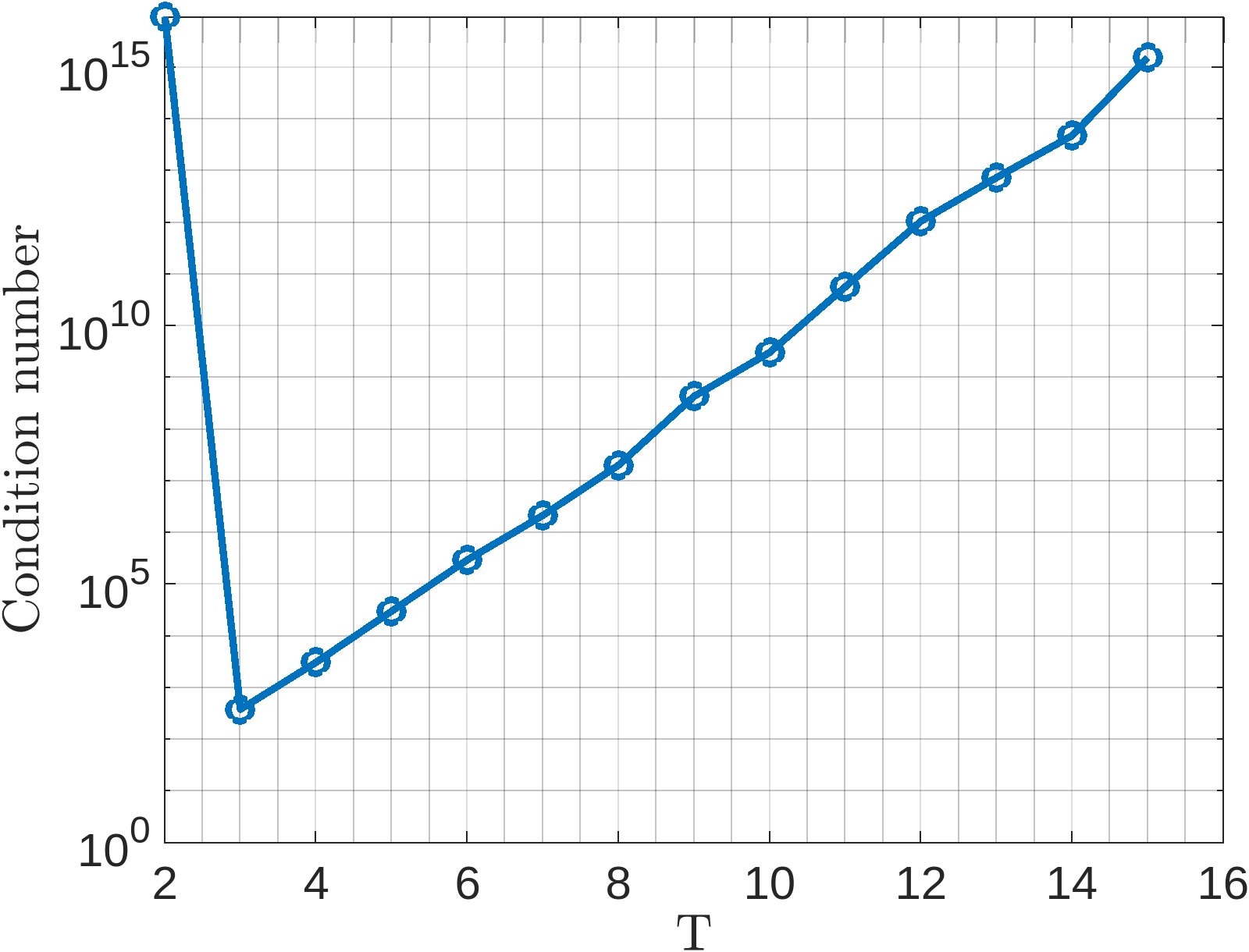}
     \caption{Condition number for different T}
     \label{fig:condition number}
 \end{figure} 
 \begin{figure}[ht]
    \centering
    \includegraphics[width=0.6\linewidth]{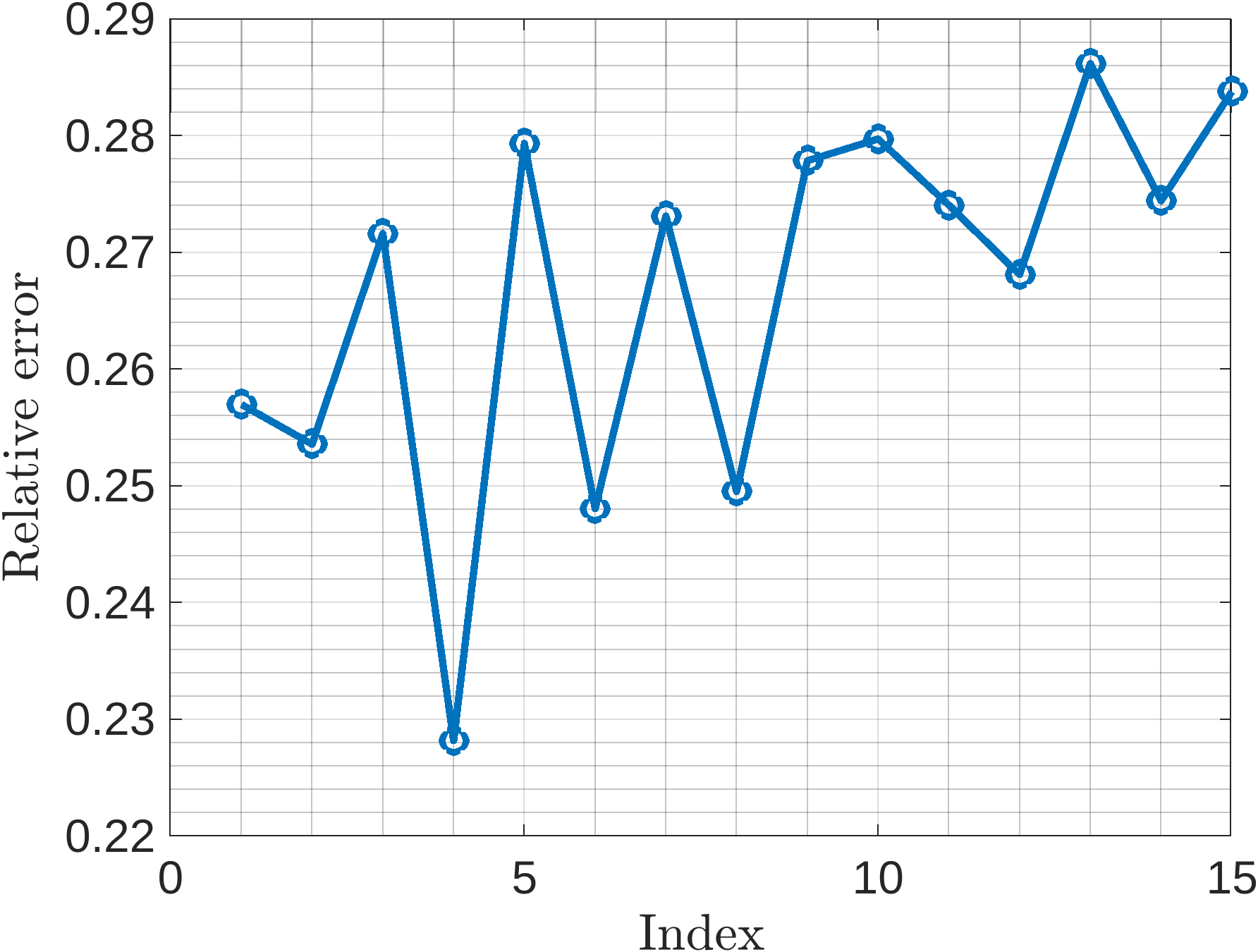 }
    \caption{Undersampled second dimension}
    \label{fig:verify conj1}

    \vspace{0.5cm} 

    \includegraphics[width=0.6\linewidth]{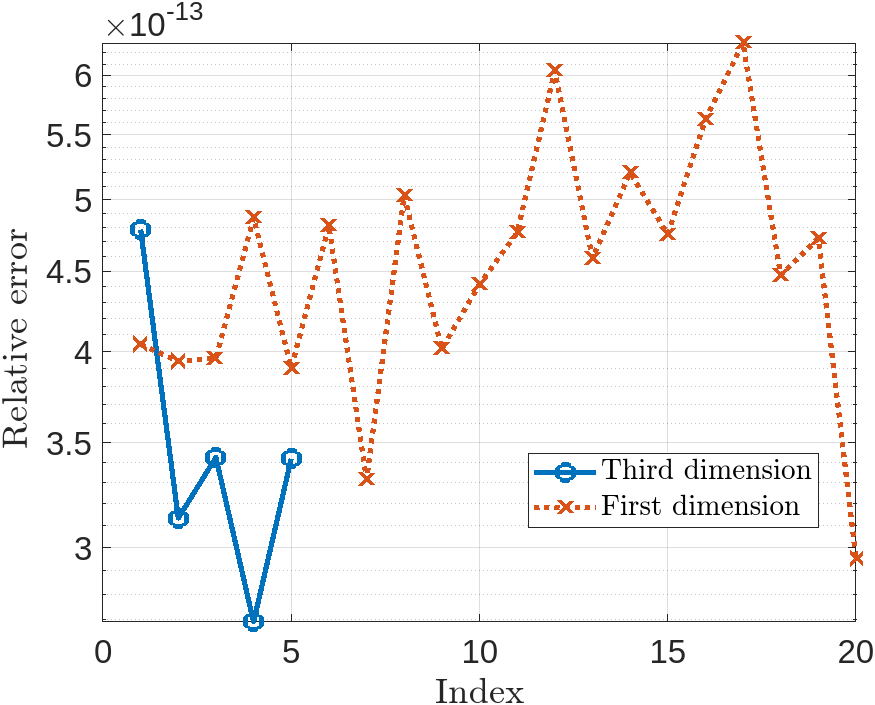} 
    \caption{Undersampled first and third dimension}
    \label{fig:verify conje}
\end{figure}




\subsubsection{Verifying the conjecture}
\label{sec 234}
In this section, we conducted simulations to evaluate \Cref{conj1}. The initial signal is a tensor of size \( 20 \times 15 \times 5 \). To test the conjecture, we first set the sampling rate to \( \alpha = 1 \), generating the sampling location set \( \Omega = \bigcup_{j=1}^{15} \Omega_{j} \), where \( \Omega_{j} \) represents the locations corresponding to the second index equal to \( j \). In each experiment, one \( \Omega_{j} \) was excluded by varying \( j \) from 1 to 15, resulting in the sampling location set \( \Omega^{j} = \Omega \setminus \Omega_{j} \). This configuration led to a sampling rate of 93\%. However, as shown in \Cref{fig:verify conj1}, the relative error remained above 0.2, indicating that the initial signal could not be fully recovered.

To further explore this, we generated a new sampling location set with \( \alpha = 0.5 \), and similarly adjusted the set by excluding all locations where the first (or third) index equals \( j \), with \( j \) varying from 1 to 20 (or from 1 to 5 for the third index). As shown in \Cref{fig:verify conje}, excluding locations where the first or third index takes these values did not affect the recovery, and the relative error reduced to \( 10^{-11} \).

\small

\bibliographystyle{plain}
\bibliography{ref}
\end{document}